\documentclass[journal,twocolumn]{IEEEtran}

\usepackage{amsmath,amsfonts,amssymb,amsthm,mathtools,bm}
\usepackage{cite}
\usepackage{graphicx}
\usepackage{xcolor}
\usepackage{booktabs}
\usepackage{algorithm}
\usepackage{algpseudocode}

\newtheorem{theorem}{Theorem}

\title{Learning Representative Port Subsets for Efficient Fluid Antenna Systems}
\author{Quan Zhou and Ming Zeng%
\thanks{
Q. Zhou is with the National University of Singapore (e-mail: zhouquan19971@gmail.com).
M. Zeng is with the Université Laval (corresponding author; e-mail: ming.zeng@gel.ulaval.ca).
}}

\begin{document}
\maketitle

\begin{abstract}
Fluid antenna systems (FASs) provide many candidate ports, but probing all ports for each channel realization incurs substantial online overhead. We study a complementary offline problem: compress the candidate-port space once and reuse the compressed set
across future channel realizations. The method collects optimal ports from sampled training realizations, requiring only one scan of the candidate ports per realization. 
On a DeepMIMO 7 × 7 FAS with 49 candidate ports, 120 training realizations yield 27 representative ports on average while retaining 99.995\% of the channel power gain achieved by exhaustive port selection.
\end{abstract}

\begin{IEEEkeywords}
Fluid antenna systems, port selection, spatial correlation, candidate-port compression, DeepMIMO
\end{IEEEkeywords}

\section{Introduction}

A fluid antenna system (FAS) is a wireless communication system in which an antenna can switch between multiple candidate ports within a compact physical space to select the position with the most favorable channel conditions~\cite{New2024ATO}.
However, selecting the best port from a large candidate set may require probing or estimating the channel quality at many ports, since the observation of each port incurs switching or signaling overhead~\cite{chai2022port}.

Existing approaches mainly reduce the cost of \emph{instantaneous} port
selection. For example, Chai \emph{et al.}~\cite{chai2022port} infer the optimal port from a limited number of observations, while other methods reduce the
effective selection cost through compact arrays, multi-port reception,
probabilistic selection, or channel prediction
\cite{wong2024compact,gonzalez2026slow,chen2025improved,han2026semi,xu2025toward}.
These methods reduce the number of channel observations needed to make a selection, but the underlying candidate-port space remains the full set of \(N\) physical ports for every new realization. 

We consider a complementary form of overhead reduction,
\emph{offline candidate-port compression}: we reduce the number of physical ports that need to be considered during online channel selection.
Suppose that a FAS is deployed in an environment whose future channel realizations follow a recurring distribution. Instead of repeatedly probing all $N$ ports, we ask whether a fixed subset of physical port positions can be learned once and reused across future channel realizations. 
This is distinct from constructing a shortlist separately for each
realization~\cite{waqar2025turbocharging}.

The key observation is that, although the optimal port depends on the channel realization, some physical ports can be optimal much more often than others. 
Moreover, because the FAS aperture is compact, spatially nearby ports tend to exhibit similar channel responses. Consequently, a port that is optimal for one channel realization can also provide near-optimal quality under other realizations whose optimal ports are nearby.
Together, these effects suggest that a small reusable set can perform well for future realizations without containing every exact optimum.

Motivated by this spatial structure, we learn a reusable subset by collecting the optimal port from each offline training realization. The resulting method requires only one scan of the candidate ports per training realization and admits a proxy-based approximation bound. This is a new FAS problem formulation and a communication-specific use of earlier work~\cite{zhou2025representative}, enabled by spatial structure of a compact FAS aperture.

The contributions are summarized as follows. 
\begin{itemize}
\item \textbf{Problem formulation}: motivated by the observation that nearby ports exhibit correlated channel qualities, we formulate a new problem of learning a reusable subset of candidate ports that preserves the performance of exhaustive port selection. 

\item \textbf{Method and theory}: we propose
a simple representative-port learning (RPL) method and establish an approximation
guarantee that quantifies the effects of proxy approximation and finite-sample coverage of proxy classes under the deployment distribution.
For fixed online observation budgets, we introduce two budgeted RPL variants.

\item \textcolor{black}{\textbf{Empirical validation}: on a DeepMIMO \cite{alkhateeb2019deepmimo} ray-tracing environment, 
RPL achieves near-exhaustive utility with a small subset, outperforms the mean-based baseline, and remains competitive with empirical greedy at substantially lower offline computational cost.
}
\end{itemize}

\section{Representative-Port Learning}
\label{sec:learning}

\textcolor{black}{
We consider a single-user single-input single-output (SISO) system with $N$ switchable candidate
ports deployed over a compact two-dimensional aperture with half-wavelength port spacing, $d=\lambda/2$.
For clarity, we develop the formulation for single-port selection; multiple-port activation can be handled analogously by treating each feasible set of simultaneously active ports as a single configuration.}
Key notation is summarized in Table~\ref{tab:notation}.

\begin{table}[t]
\centering
\caption{Key notation.}
\label{tab:notation}
\begin{tabular}{cc}
\hline
Symbol & Meaning \\
\hline
$N$ & Number of candidate ports in the FAS \\
$S$ & Number of offline training channel realizations \\
$M$ & Number of test channel realizations \\
$\mathcal A$ & \textcolor{black}{Generic subset of candidate ports} \\
$\mathcal A_S$ & \textcolor{black}{RPL subset learned from $S$ training realizations} \\
$R$ & Number of ports in the RPL subset, $R=|\mathcal A_S|$ \\
$K$ & Target online observation budget \\
\hline
\end{tabular}
\end{table}

\subsection{Reusable Port-Subset Formulation}
\label{sec:formulation}

Let $\mathcal P$ denote the distribution of channel realizations encountered during deployment. Let $\theta$ denote a channel
realization associated with a receiver location and its propagation
conditions, and assume
\[
\theta \sim \mathcal P.
\]
Let $[N]:=\{1,\ldots,N\}$ index the physical port positions within the FAS aperture,
so the same index refers to the same relative port position across channel realizations.

For a channel realization $\theta$, let $h_\theta(n)\in\mathbb C$ denote the
channel coefficient at port $n$, while $U_\theta(n)$ be the instantaneous
utility of activating port $n$. In the experiments we use $U_\theta(n):=|h_\theta(n)|^2$, which is proportional to the received \textcolor{black}{signal-to-noise ratio} under a fixed transmit power and noise level; the formulation also applies to other communication utilities. The exhaustive-selection optimum is
\begin{equation}
    U_\theta^\star
    :=
    \max_{n\in[N]} U_\theta(n).
    \label{eq:optimal-port}
\end{equation}

Instead of probing all $N$ ports online, we seek a fixed subset
$\mathcal A\subseteq[N]$ learned offline. For a new channel realization, online
selection is restricted to $\mathcal A$, giving the utility loss
\begin{equation}
    \operatorname{Regret}_{\mathcal A}(\theta)
    :=
    U_\theta^\star
    -
    \max_{m\in\mathcal A} U_\theta(m).
    \label{eq:instance-regret}
\end{equation}
Let $|\mathcal A|$ be the cardinality. The offline objective is therefore
\begin{equation}
    \min_{\mathcal A\subseteq[N]}
    \mathbb E_{\theta\sim\mathcal P}
    \left[
    \operatorname{Regret}_{\mathcal A}(\theta)
    \right],
    \qquad |\mathcal A|\ll N.
    \label{eq:subset-objective}
\end{equation}

\subsection{Algorithm and Online Use}

\textcolor{black}{To construct a solution to the subset-selection problem in \eqref{eq:subset-objective} we propose RPL.}
Given $S$ offline training channel realizations $\theta_1,\ldots,\theta_S\overset{\mathrm{i.i.d.}}{\sim}\mathcal P$, it finds the optimal port for each realization and takes the union of these ports:
\begin{equation}
    n^\star(\theta_i)
    :=
    \arg\max_{n\in[N]} U_{\theta_i}(n),
    \qquad
    \mathcal A_S
    :=
    \bigcup_{i=1}^{S}\{n^\star(\theta_i)\}.
    \label{eq:representative-set}
\end{equation}
The procedure is summarized in Algorithm~\ref{alg:representative-subset}.
RPL only requires the optimal port index for each training realization; the index may be obtained by exhaustive probing or by other port-identification methods.

\begin{algorithm}[t]
\caption{Representative-Port Learning (RPL)}
\label{alg:representative-subset}
\begin{algorithmic}[1]
\Require $S$ training channel realizations, candidate ports $[N]$
\Ensure RPL subset $\mathcal A_S$
\State $\mathcal A_S\gets\varnothing$
\For{$i=1,\ldots,S$}
    \State     
    Evaluate the utility at all ports for realization $\theta_i$
    \State $n^\star(\theta_i)\gets\arg\max_{n\in[N]}U_{\theta_i}(n)$
    \State $\mathcal A_S\gets\mathcal A_S\cup\{n^\star(\theta_i)\}$
\EndFor
\State \Return $\mathcal A_S$
\Statex 
\textcolor{black}{
\footnotesize\textit{Note:} Under $\theta_1,\ldots,\theta_S\overset{\mathrm{i.i.d.}}{\sim}\mathcal P$, both $\mathcal A_S$ and $R=|\mathcal A_S|$ are random.}
\end{algorithmic}
\end{algorithm}

After offline learning, when a new channel realization $\theta$ becomes
available, we evaluate the channel quality only at the ports in the learned subset $\mathcal A_S$ and select the port with the highest utility.
We count one channel-quality sounding per evaluated port. 
\begin{equation}
    \hat n(\theta)
    =
    \arg\max_{n\in\mathcal A_S} U_\theta(n).
    \label{eq:online-selection}
\end{equation}
Thus, instead of evaluating all $N$ candidate ports, online selection evaluates only $R:=|\mathcal A_S|$ ports. 
The same physical subset $\mathcal A_S$ is reused for all future channel realizations.

\subsection{Why Representative Ports Suffice}

The optimal port can change across channel realizations, but nearby ports in a densely spaced FAS tend to have similar channel qualities (due to spatial correlation) and can therefore serve as proxies.
We partition $[N]$ into disjoint proxy classes \(\mathcal C_1,\ldots,\mathcal C_J\), such that any two ports in the same class have utility values differing by at most \(\varepsilon\) for any channel realization possible under the deployment distribution $\mathcal{P}$.
Note that these proxy classes are introduced only for analysis.
\textcolor{black}{Assume also $U_\theta(n)\in[0,B]$ for some finite $B>0$.} Let
\begin{equation}
    q_j
    :=
    \Pr_{\theta\sim\mathcal P}
    \!\left(n^\star(\theta)\in\mathcal C_j\right)
    \label{eq:class-prob}
\end{equation}
denote the probability that the optimal port lies in $\mathcal C_j$.

\begin{theorem}[Representative-port approximation]
\label{thm:proxy-bound}
For an independent channel realization $\theta$, the subset $\mathcal A_S$
obtained by Algorithm~\ref{alg:representative-subset} satisfies
\begin{equation}
\begin{aligned}
\mathbb E_{\theta,\mathcal A_S}
\!\left[\operatorname{Regret}_{\mathcal A_S}(\theta)\right]
&\leq
\varepsilon+B\sum_{j=1}^J q_j e^{-Sq_j}.
\end{aligned}
\label{eq:proxy-bound}
\end{equation}
\end{theorem}

Here, the expectation over $\mathcal A_S$ is taken because the learned subset
depends on the random offline training realizations.

\begin{proof}
Condition on $n^\star(\theta)\in\mathcal C_j$. 
If at least one training optimum lies in \(C_j\), then some \(a\in\mathcal A_S\cap C_j\) satisfies \(U_\theta^\star-U_\theta(a)\le\varepsilon\) by the proxy-class property;
otherwise the regret is at most $B$. The probability of the latter event is
$q_j(1-q_j)^S$. Summing over $j$, and using $(1-q_j)^S\leq e^{-Sq_j}$ give the bound.
\end{proof}
The theorem shows why exact optimal-port recovery is unnecessary: it is sufficient to recover any representative from the optimal port's proxy class, which incurs at most \(\varepsilon\) utility loss, while the probability of missing a class with probability \(q_j\) decays exponentially with \(S\).

\subsection{Computational Complexity}
\label{sec:complexity}

RPL is deliberately lightweight: it requires neither model fitting nor iterative subset optimization, and once the optimal port is identified for each offline realization, only its index needs to be retained. Under exhaustive optimal-port identification, this entails \(O(SN)\) offline utility evaluations, while this cost can be further reduced by combining RPL with alternative low-overhead identification mechanisms.
After deployment, RPL reduces the number of online port evaluations from \(O(N)\) to \(O(R)\). \textcolor{black}{The corresponding costs of the baselines are summarized in Table~\ref{tab:complexity}.}

\begin{table}[t]
  \centering
  \caption{Offline complexity and retained online candidate size for
  the subset-selection methods. Here $R=|\mathcal A_S|$ \textcolor{black}{denotes the (random) size of RPL subset}, while $K$ is the target online budget.}
  \label{tab:complexity}
  \begin{tabular}{lcc}
    \toprule
    Method & Offline cost & Online candidates \\
    \midrule
    RPL & $O(SN)$ & $R$ \\
    G-BRPL & $O(SN+SRK)^{*}$ & $\min\{R,K\}$ \\
    F-BRPL & $O(SN)$ & $\min\{R,K\}$ \\
    Top-$K$ Mean Utility & $O(SN)$ & $K$ \\
    Empirical Greedy & $O(SNK)$ & $K$ \\
    Uniform & $O(NK)$ & $K$ \\
    \bottomrule
\multicolumn{3}{l}{\footnotesize
* When $K<R$; otherwise, G-BRPL reduces to RPL.}
  \end{tabular}
\end{table}

\section{Numerical Results}
\label{sec:evaluation}

\textcolor{black}{In this section,} we evaluate whether representative-port learning can compress the online candidate-port space while preserving exhaustive port-selection utility.
All experiments use the same DeepMIMO~\cite{alkhateeb2019deepmimo} ray-tracing environment and a compact \(7\times7\) FAS with $N=49$ switchable ports.

\subsection{Experimental Setup}

Each channel realization \(\theta\) is constructed from the propagation environment at a reference receiver location \(p_\theta\).
The local $7\times7$ FAS is centered at $p_\theta$, and the channel coefficients at its candidate ports are synthesized from the dominant propagation paths at this reference location.

\paragraph{Receiver-location distribution}
Each channel realization $\theta$ is associated with a reference receiver location $p_\theta$. We define the deployment distribution by assigning
sampling probabilities to the valid DeepMIMO receiver locations. 
\textcolor{black}{Specifically,
we evaluate a four-hotspot Gaussian-mixture density at each receiver location
to model a spatially nonuniform deployment, motivated by clustered user
distributions in wireless networks~\cite{de2024nonuniform}.}
The probability of a location $\mathbf z=(z_x,z_y)$ is proportional to
\[
    f(\mathbf z)\propto
    \sum_{r=1}^{4}w_r
    \exp\!\left[
        -\frac12\left(
        \frac{(z_x-x_r)^2}{\sigma_{x,r}^2}
        +
        \frac{(z_y-y_r)^2}{\sigma_{y,r}^2}
        \right)
    \right],
\]
where the four hotspots are characterized by their centers $(x_r,y_r)$, spatial spreads $(\sigma_{x,r},\sigma_{y,r})$ along the two coordinate axes, and relative sampling weights $w_r$. 
For $r=[4]$, the corresponding parameter values are
$(0.20,0.22,0.025,0.10,0.25)$,
$(0.72,0.72,0.025,0.10,0.20)$,
$(0.14,0.48,0.020,0.08,0.20)$, and
$(0.46,0.60,0.020,0.08,0.20)$.

\paragraph{Channel and port model}
We use the DeepMIMO v4 \texttt{O1\_28} scenario with transmitter set~5
(TX index~0) and receiver set~2 (UG3). We use the channel on subcarrier~0
for a SISO link and do not model Doppler-induced time variation.
Among the valid receiver locations, $487{,}711$ reference points are available for sampling.

For each reference receiver location $p_{\theta}$, we synthesize the channel across the $7\times 7$ FAS using the strongest $L=20$ propagation paths. 
\textcolor{black}{Under the far-field plane-wave approximation~\cite{ramirez2024spatial}, the channel at port $n$ is modeled as}
\begin{equation}
    h_\theta(n)
    =
    \sum_{\ell=1}^{L}
    a_{\theta,\ell}
    \exp\!\left(
        -\mathrm{i}\frac{2\pi}{\lambda}
        \mathbf u_{\theta,\ell}^{\mathsf T}\Delta\mathbf r_n
    \right),
    \qquad n\in[N],
    \label{eq:deepmimo-channel}
\end{equation}
where $a_{\theta,\ell}\in\mathbb C$ is the complex gain of path $\ell$
at the reference location, $\mathbf u_{\theta,\ell}\in\mathbb R^2$ is its
in-plane propagation direction, and
$\Delta\mathbf r_n\in\mathbb R^2$ is the physical offset of port $n$
from $p_\theta$. 
The $49$ port form a centered square grid with spacing $d=\lambda/2$:
\begin{equation*}
    \Delta\mathbf r_n
    =
    \begin{bmatrix}
        (p-4)d\\
        (q-4)d
    \end{bmatrix},
    \quad
    n=7(p-1)+q,\quad p,q\in[7].
    \label{eq:fas-grid}
\end{equation*}
\textcolor{black}{This model implies spatially smooth channel variation; the corresponding Lipschitz bound is given in the Appendix~\ref{sec:appendix}.}

\subsection{Evaluation Metrics and Baselines}

All subset-selection methods are evaluated on the same
$M=1000$ i.i.d. test channel realizations
$\theta_1,\ldots,\theta_M\sim\mathcal P$.

\paragraph{Metrics}
Given $M$ i.i.d. test channel realizations
$\theta_1,\ldots,\theta_M\sim\mathcal P$, our primary metric for a subset $\mathcal A\subseteq[N]$ is the relative utility loss
\begin{equation}
    \mathcal L_{\rm rel}(\mathcal A)
    =
    1-
    \frac{1}{M}
    \sum_{j=1}^M
    \frac{\max_{n\in\mathcal A}U_{\theta_j}(n)}{\max_{n\in[N]}U_{\theta_j}(n)}.
    \label{eq:relative-loss}
\end{equation}
We also report near-optimal coverage at threshold $\tau>0$
\begin{equation}
    \mathcal C_\tau(\mathcal A)
    =
    \frac{1}{M}\sum_{j=1}^{M}
    \mathbf 1\!\left\{
        \max_{n\in\mathcal A}U_{\theta_j}(n)\ge \tau \max_{n\in[N]}U_{\theta_j}(n)
    \right\}    .
    \label{eq:near-optimal-coverage}
\end{equation}
The online probing fraction is $K/N$.

\paragraph{Budgeted RPL variants}
\textcolor{black}{Since the learned subset \(\mathcal{A}_S\) is random, its cardinality \(R=|\mathcal{A}_S|\) is also random. Thus, for a prescribed online observation budget \(K\), the event \(K<R\) may occur. In this case, we consider two budgeted RPL variants that further trim \(\mathcal{A}_S\) to \(K\) ports.} No port outside \(A_S\) is introduced.

\textit{Frequency-based Budgeted RPL (F-BRPL)} selects the \(K\) ports in \(\mathcal A_S\) with the largest counts of being the exhaustive-optimal port across the training realizations. 

\textit{Greedy Budgeted RPL (G-BRPL)} starts with an empty subset and iteratively adds the port in \(\mathcal A_S\) that yields the largest marginal increase in the empirical objective 
\begin{equation}
\widehat U_S(\mathcal A)
=
\frac{1}{S}\sum_{i=1}^S
\max_{n\in\mathcal A} U_{\theta_i}(n),\quad \widehat U_S(\emptyset):=0.
\label{equ:empirical-objective}
\end{equation}

\paragraph{Reusable subset-selection baselines}
We compare RPL with three baselines that construct a fixed $K$-port subset from the same training set.

\textit{Uniform} uses deterministic farthest-point sampling on the
physical port coordinates: starting from a fixed initial port, it iteratively
selects the port that is farthest from the already selected set,
thereby providing a spatially distributed subset.

\textit{Top-$K$ Mean Utility} selects the $K$ ports with the largest
average channel utility over the available training realizations.

\textit{Empirical Greedy} starts with an empty port subset and
iteratively adds one port at a time. At each step, for each remaining
candidate port, it computes the average increase in empirical objective $\widehat U_S(\mathcal A)$ over the training set when that port is added to the current subset, and selects the port with the largest increase. The process continues until $K$ ports are selected. 
This is the standard greedy strategy for cardinality-constrained submodular maximization \cite{nemhauser1978analysis}.

\paragraph{Limited-observation baseline}
We additionally compare against a Chai-style limited-observation LSTM \cite{chai2022port}. This baseline probes exactly $R$ approximately equispaced ports using the same farthest-point rule as Uniform, predicts the full $N$-dimensional utility vector, and then selects a port using the predicted utilities. Thus both methods use the same number of observed ports, but RPL permanently compresses the candidate space whereas the LSTM performs inference over all $N$ ports.

\paragraph{Computational comparison}
Table~\ref{tab:complexity} summarizes the computational costs. RPL, F-BRPL, and Top-\(K\) Mean Utility require \(O(SN)\) offline operations, compared with \(O(SNK)\) for Empirical Greedy and \(O(NK)\) for Uniform. G-BRPL has the same \(O(SN)\) complexity as RPL when \(R\le K\), but requires \(O(SN+SRK)\) when \(R>K\) due to the extra trimming step. Online, RPL evaluates \(R\) ports, whereas the Chai-style LSTM requires model training and port inference.

\subsection{Utility--Overhead Tradeoff}

We compare the fixed-budget RPL variants and other subset-selection methods on a shared test set of $M=1000$ realizations, using $20$ repetitions with fresh training sets of $S=200$ realizations and budgets $K\in\{5,10,15,20,25,27,28,30\}$. One repetition yielded $R=26$; we exclude it for $K\in\{27,28,30\}$ to ensure a fair comparison at the same budget.

Fig.~\ref{fig:utility-overhead} summarizes the fixed-budget results over $K$.
The upper plot shows that the relative utility loss decreases for all methods as the online budget \(K\) increases. 
\textcolor{black}{G-BRPL closely tracks Empirical Greedy while avoiding the latter's \(O(SNK)\) iterative subset search.
Uniform is slightly better at the smallest budgets, but both budgeted RPL variants become increasingly competitive as the budget grows. F-BRPL incurs slightly higher utility loss than G-BRPL, but offers a modest computational advantage. Top-(K) Mean Utility consistently incurs substantially higher loss.}

The lower plot provides a complementary view of how frequently the selected subset achieves at least \(99.5\%\) of the exhaustive-selection utility.
G-BRPL coverage increases from $0.956$ to $1.0$ as $K$ increases. 
Empirical Greedy closely tracks RPL, whereas Top-K Mean Utility remains around
$0.877$--$0.899$ even at $K=25$.

\begin{figure}[t]
    \centering
    \includegraphics[width=0.8\columnwidth]{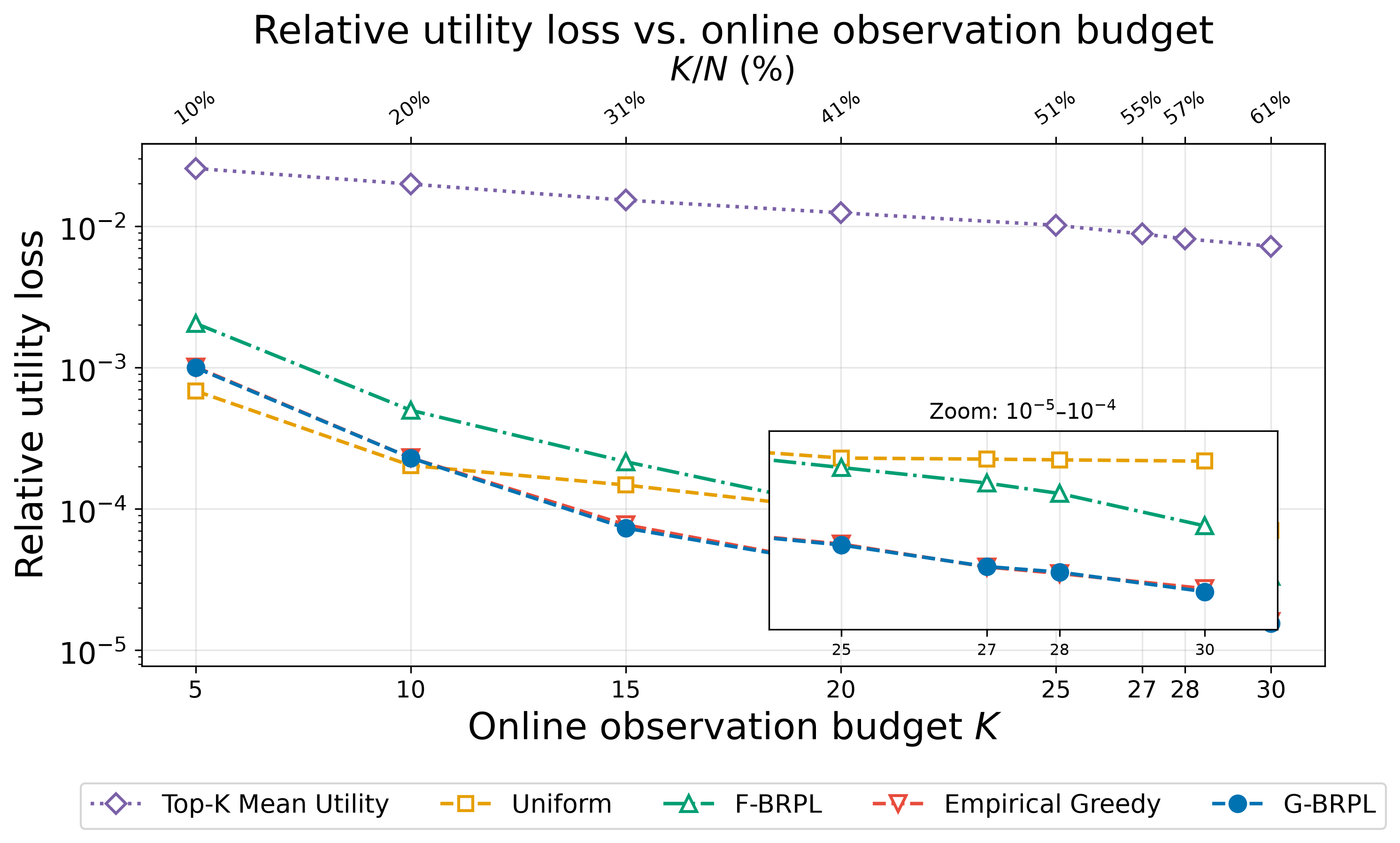}
    
    \vspace{2pt}
    
    \includegraphics[width=0.8\columnwidth]{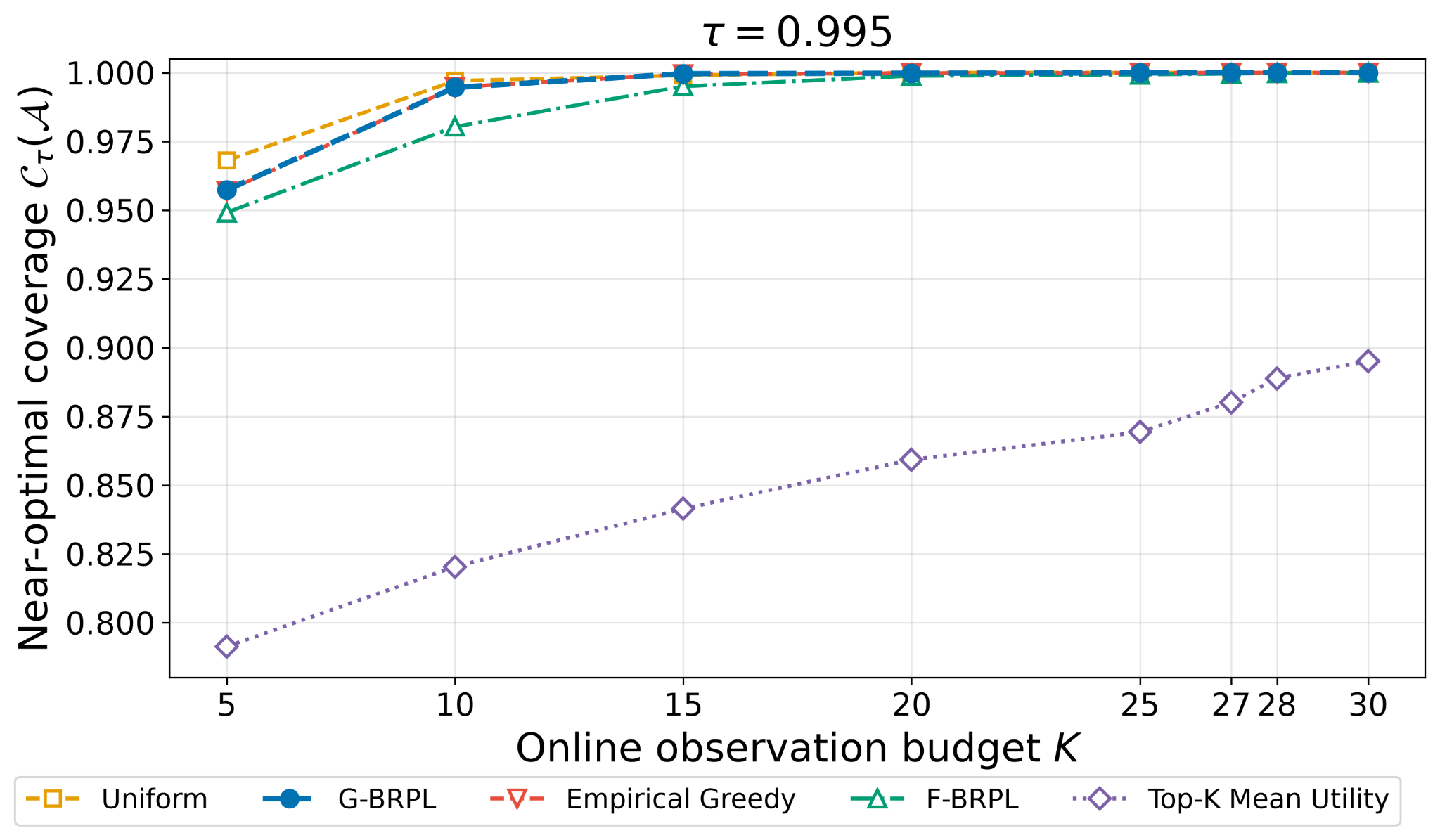}
    \caption{Fixed-budget utility--overhead tradeoff: relative utility loss (upper) and near-optimal coverage (lower)
    versus the online budget $K$.}
    \label{fig:utility-overhead}
\end{figure}

\subsection{Limited-Observation Comparison}

We next compare RPL with a Chai-style limited-observation LSTM.
For each $S\in\{40,80,120,160,200,400,600\}$, we perform $10$ repetitions using the same $S$ training realizations for both methods. RPL uses all $S$ realizations to construct its representative subset,
while the LSTM uses $80\%$ for model fitting and $20\%$ for validation, with a randomized split.
The LSTM uses a one-layer architecture with $32$ hidden units and is trained using MSE and Adam (learning rate \(10^{-3}\), batch size $32$, maximum $500$ epochs, and early-stopping patience $50$).

At test time, the RPL subset cardinality $R$ determines the common observation budget $K$. RPL selects the port with the highest observed utility among its $R$ representative ports.
In contrast, the LSTM uses the $R$ observed approximately equispaced ports to predict the utilities of all $N$ ports and selects the port with the highest predicted utility. 

The two methods use the same number of online observations but serve
different purposes: RPL compresses the candidate space, whereas the LSTM
infers the remaining utilities for each realization. Thus, we compare
them at a common online observation budget rather than an identical
offline training protocol.

\begin{figure}[t]
  \centering
  \includegraphics[width=0.75\columnwidth]{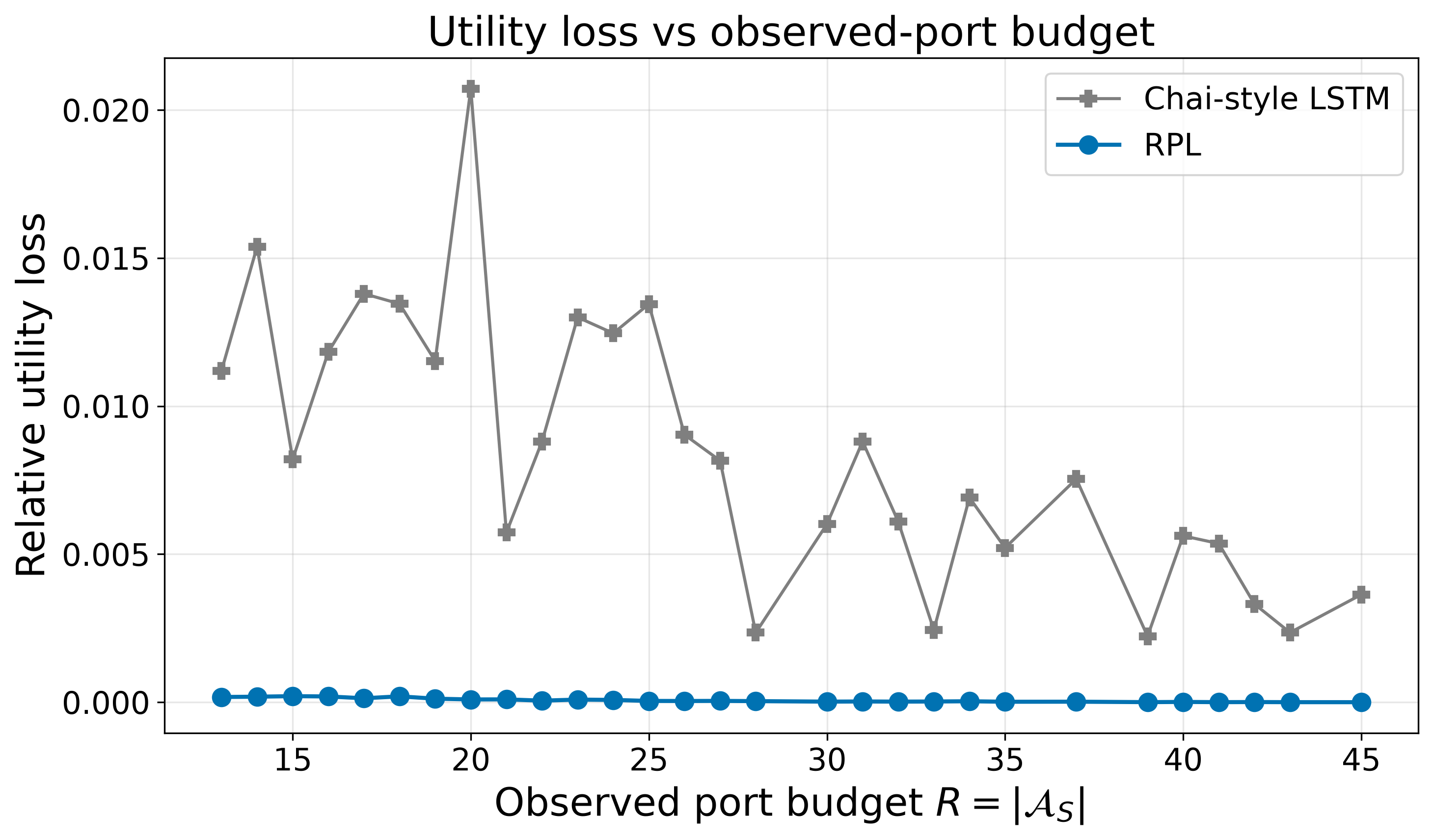}
  \caption{Relative utility loss versus the common observed-port budget $R$
  for RPL and the Chai-style limited-observation LSTM. Each method uses the same number of observed ports.}
  \label{fig:lstm}
\end{figure}

\subsection{Offline Sample Efficiency}

We next study the offline sample efficiency of RPL by varying
$S\in\{40,80,120,160,200,300,400,600\}$ for $20$ repetitions. Each repetition samples $S$
realizations without replacement from the $600$-realization training pool.
The learned subset size $R$ is then used for all subset-selection baselines (set $K=R$) and evaluated on the same $M=1000$ test set.
\textcolor{black}{Under this setup, both F-BRPL and G-BRPL produce the same subset as RPL. We therefore omit them from this experiment.}

Table~\ref{tab:ssweep} and Fig.~\ref{fig:ssweep} summarize the results.
As \(S\) increases, RPL progressively enlarges its learned subset and its relative utility loss decreases.
In particular, at $S=120$, RPL retains $27.1\pm2.9$ ports on average and incurs a relative utility loss of only $4.64\times10^{-5}$.
Empirical Greedy still closely tracks RPL, while Top-K Mean Utility consistently incurs substantially larger loss. Uniform is competitive at small \(S\), but the performance becomes less favorable as the available training data increase.

\begin{table}[t]
\centering
\caption{Offline sample-efficiency sweep: mean learned subset size and
relative utility loss (mean $\pm$ sample std over $20$ repetitions).}
\label{tab:ssweep}
\resizebox{\columnwidth}{!}{%
\begin{tabular}{c c c c c c}
\toprule
$S$ & $R=|\mathcal A_S|$ & RPL & Uniform & Top-K Mean Utility & Empirical Greedy\\
\midrule
40  & $15.20\!\pm\!2.17$ & $(3.07\!\pm\!3.85)\!\times\!10^{-4}$ & $(1.54\!\pm\!0.21)\!\times\!10^{-4}$ & $(1.61\!\pm\!0.37)\!\times\!10^{-2}$ & $(4.57\!\pm\!8.36)\!\times\!10^{-4}$\\
80  & $22.65\!\pm\!2.06$ & $(4.26\!\pm\!15.73)\!\times\!10^{-4}$ & $(9.50\!\pm\!1.34)\!\times\!10^{-5}$ & $(1.13\!\pm\!0.29)\!\times\!10^{-2}$ & $(4.36\!\pm\!15.73)\!\times\!10^{-4}$\\
120 & $27.10\!\pm\!2.86$ & $(4.64\!\pm\!2.06)\!\times\!10^{-5}$ & $(7.43\!\pm\!1.30)\!\times\!10^{-5}$ & $(8.77\!\pm\!1.74)\!\times\!10^{-3}$ & $(4.88\!\pm\!3.03)\!\times\!10^{-5}$\\
160 & $31.40\!\pm\!2.21$ & $(2.62\!\pm\!1.29)\!\times\!10^{-5}$ & $(5.87\!\pm\!1.22)\!\times\!10^{-5}$ & $(6.74\!\pm\!1.20)\!\times\!10^{-3}$ & $(2.38\!\pm\!0.97)\!\times\!10^{-5}$\\
200 & $32.80\!\pm\!2.63$ & $(2.35\!\pm\!1.41)\!\times\!10^{-5}$ & $(5.44\!\pm\!1.19)\!\times\!10^{-5}$ & $(6.08\!\pm\!1.02)\!\times\!10^{-3}$ & $(2.30\!\pm\!1.50)\!\times\!10^{-5}$\\
300 & $38.45\!\pm\!1.99$ & $(5.35\!\pm\!6.24)\!\times\!10^{-6}$ & $(3.83\!\pm\!0.28)\!\times\!10^{-5}$ & $(3.93\!\pm\!0.70)\!\times\!10^{-3}$ & $(3.52\!\pm\!3.01)\!\times\!10^{-6}$\\
400 & $41.15\!\pm\!1.50$ & $(2.42\!\pm\!1.86)\!\times\!10^{-6}$ & $(3.40\!\pm\!0.71)\!\times\!10^{-5}$ & $(3.10\!\pm\!0.51)\!\times\!10^{-3}$ & $(2.12\!\pm\!1.64)\!\times\!10^{-6}$\\
600 & $45.00\!\pm\!0.00$ & $(4.77\!\pm\!0.00)\!\times\!10^{-7}$ & $(1.58\!\pm\!0.00)\!\times\!10^{-5}$ & $(1.79\!\pm\!0.00)\!\times\!10^{-3}$ & $(4.77\!\pm\!0.00)\!\times\!10^{-7}$\\
\bottomrule
\end{tabular}}
\end{table}

\begin{figure}[t]
  \centering
  \includegraphics[width=\columnwidth]{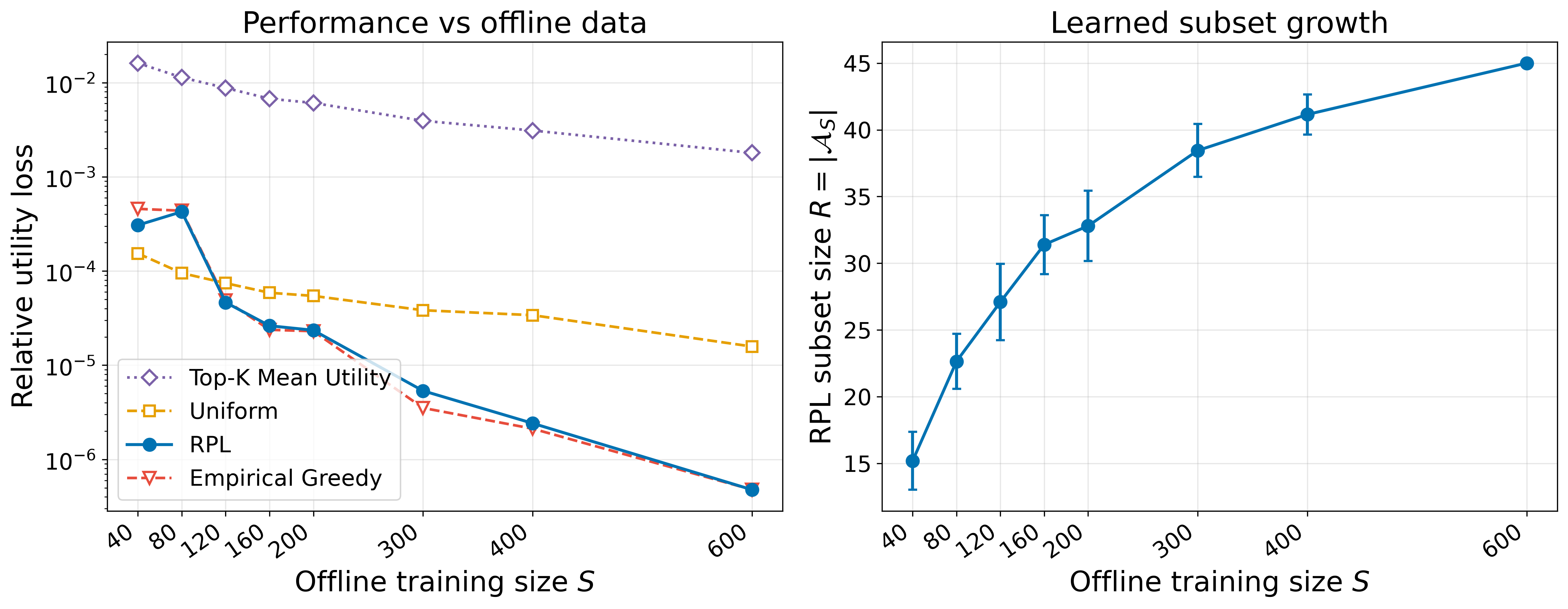}
\caption{Offline sample efficiency under the deployment distribution.
Left: relative utility loss versus training size $S$, with all methods using an online observation budget $K=R$; right: RPL subset cardinality $R$ versus $S$, with mean $\pm$ one standard deviation over $20$ repetitions.}
  \label{fig:ssweep}
\end{figure}

\section{Conclusion}

We introduced a new approach to reducing port-selection overhead in FASs by learning a reusable subset of representative physical ports offline and restricting subsequent online selection to this subset. The approach exploits a basic property of compact FAS
apertures: spatially nearby ports tend to exhibit similar channel
utilities, so a representative port can provide near-optimal utility
even when it is not the exact exhaustive optimum. 
This structure leads
to a simple learning procedure that requires only the index of the
optimal port from each training channel realization, with no iterative subset
optimization or model fitting. 

Experiments on a DeepMIMO-based FAS demonstrate that RPL substantially reduces the candidate-port space with negligible utility loss. Under online observation budgets, G-BRPL achieves performance comparable to empirical greedy at substantially lower computational cost, while F-BRPL is computationally cheaper but incurs slightly higher utility loss.

RPL can be combined with methods that infer the best port from limited channel observations to further reduce online probing. 
A natural extension is to learn contextual representative subsets that adapt to changing conditions.

\appendix

\subsection{Spatial Smoothness and Proxy Property}
\label{sec:appendix}
Under the far-field model, the channel construction also gives the spatial proxy property used in Theorem~\ref{thm:proxy-bound}. 
Let
$A_\theta:=\sum_{\ell=1}^L|a_{\theta,\ell}|\le\bar A$ \textcolor{black}{for some finite $\bar A>0$} and assume
$\|\mathbf u_{\theta,\ell}\|_2=1$. Since
$|e^{\mathrm{i}x}-e^{\mathrm{i}y}|\le |x-y|$,
\[
|h_\theta(n)-h_\theta(m)|
\le
\frac{2\pi A_\theta}{\lambda}
\|\Delta\mathbf r_n-\Delta\mathbf r_m\|_2.
\]
Moreover, $|h_\theta(n)|\le A_\theta$, and hence
\[
|U_\theta(n)-U_\theta(m)|
\le
\frac{4\pi A_\theta^2}{\lambda}
\|\Delta\mathbf r_n-\Delta\mathbf r_m\|_2.
\]
Therefore, if $A_\theta\le\bar A$ over the considered channel class,
\begin{equation*}
|U_\theta(n)-U_\theta(m)|
\le
\frac{4\pi\bar A^2}{\lambda}
\|\Delta\mathbf r_n-\Delta\mathbf r_m\|_2.
\end{equation*}

Thus, nearby physical ports can serve as low-loss proxies even when they are not exact exhaustive optima.
In particular, if every proxy class has physical diameter at most \(\delta\), the Lipschitz bound above gives \(\epsilon \le (4\pi\bar A^2/\lambda)\delta\).

\bibliographystyle{IEEEtran}
\bibliography{ref}

@article{New2024ATO,
  title={A Tutorial on Fluid Antenna System for 6G Networks: Encompassing Communication Theory, Optimization Methods and Hardware Designs},
  author={Wee Kiat New and Kai-Kit Wong and Hao Xu and Chao Wang and Farshad Rostami Ghadi and Jichen Zhang and Junhui Rao and Ross D. Murch and Pablo Ram{\'i}rez-Espinosa and David Morales-Jim{\'e}nez and Chan-Byoung Chae and Kin-Fai Tong},
  journal={IEEE Communications Surveys \& Tutorials},
  year={2024},
  volume={27},
  pages={2325-2377},
  url={https://api.semanticscholar.org/CorpusID:271039289}
}

@article{nemhauser1978analysis,
  title={An analysis of approximations for maximizing submodular set functions—I},
  author={Nemhauser, George L and Wolsey, Laurence A and Fisher, Marshall L},
  journal={Mathematical programming},
  volume={14},
  number={1},
  pages={265--294},
  year={1978},
  publisher={Springer}
}

@ARTICLE{gonzalez2026slow,
  author={González-Coma, José P. and López-Martínez, F. Javier},
  journal={IEEE Wireless Communications Letters}, 
  title={Slow Fluid Antenna Multiple Access With Multiport Receivers}, 
  year={2026},
  volume={15},
  number={},
  pages={1280-1284},
  doi={10.1109/LWC.2026.3651420}
}

@article{waqar2025turbocharging,
  title={Turbocharging fluid antenna multiple access},
  author={Waqar, Noor and Wong, Kai-Kit and Chae, Chan-Byoung and Murch, Ross},
  journal={IEEE Transactions on Wireless Communications},
  year={2025},
  publisher={IEEE}
}

@ARTICLE{wong2024compact,
  author={Wong, Kai-Kit and Chae, Chan-Byoung and Tong, Kin-Fai},
  journal={IEEE Transactions on Wireless Communications}, 
  title={Compact Ultra Massive Antenna Array: A Simple Open-Loop Massive Connectivity Scheme}, 
  year={2024},
  volume={23},
  number={6},
  pages={6279-6294},
  doi={10.1109/TWC.2023.3330932}}

@article{chen2025improved,
  title={Improved joint transmit and receive port selection for capacity maximization in fluid-MIMO systems},
  author={Chen, Jung-Chieh and Cheng, Tzu-Lu and Wong, Kai-Kit and Shin, Hyundong},
  journal={IEEE Wireless Communications Letters},
  volume={14},
  number={6},
  pages={1693--1697},
  year={2025},
  publisher={IEEE}
}

@ARTICLE{chai2022port,
  author={Chai, Zhi and Wong, Kai-Kit and Tong, Kin-Fai and Chen, Yu and Zhang, Yangyang},
  journal={IEEE Communications Letters}, 
  title={Port Selection for Fluid Antenna Systems}, 
  year={2022},
  volume={26},
  number={5},
  pages={1180-1184},
  doi={10.1109/LCOMM.2022.3152451}
}

@article{xu2025toward,
  title={Toward practical fluid antenna systems: Co-optimizing hardware and software for port selection and beamforming},
  author={Xu, Sai and Wong, Kai-Kit and Du, Yanan and Hong, Hanjiang and Chae, Chan-Byoung and Liu, Baiyang and Tong, Kin-Fai},
  journal={IEEE Transactions on Wireless Communications},
  year={2025},
  publisher={IEEE}
}

@article{han2026semi,
author = {Han, Tianyu and Zhu, Yongxu and Wong, Kai-Kit and Zheng, Gan and Chae, Chan-Byoung and You, Xiaohu},
year = {2026},
month = {01},
pages = {21625-21638},
title = {Partial Fluid Antenna System: Port Selection via Statistical Analysis},
volume = {25},
journal = {IEEE Transactions on Wireless Communications},
doi = {10.1109/TWC.2026.3722171}
}

@misc{zhou2025representative,
      title={Representative Action Selection for Large Action Space Bandit Families}, 
      author={Quan Zhou and Mark Kozdoba and Shie Mannor},
      year={2026},
      eprint={2505.18269},
      archivePrefix={arXiv},
      primaryClass={cs.LG},
      url={https://arxiv.org/abs/2505.18269}, 
}

@misc{alkhateeb2019deepmimo,
      title={DeepMIMO: A Generic Deep Learning Dataset for Millimeter Wave and Massive MIMO Applications}, 
      author={Ahmed Alkhateeb},
      year={2019},
      eprint={1902.06435},
      archivePrefix={arXiv},
      primaryClass={cs.IT},
      url={https://arxiv.org/abs/1902.06435}, 
}

@INPROCEEDINGS{de2024nonuniform,
  author={Filippo, Bruno De and Ahmad, Bilal and Riviello, Daniel Gaetano and Guidotti, Alessandro and Vanelli-Coralli, Alessandro},
  booktitle={2024 IEEE International Mediterranean Conference on Communications and Networking (MeditCom)}, 
  title={Non-Uniform User Distribution in Non-Terrestrial Networks with Application to User Scheduling}, 
  year={2024},
  volume={},
  number={},
  pages={441-446},
  doi={10.1109/MeditCom61057.2024.10621409}}

@ARTICLE{ramirez2024spatial,
  author={Ramírez-Espinosa, Pablo and Morales-Jimenez, David and Wong, Kai-Kit},
  journal={IEEE Transactions on Wireless Communications}, 
  title={A New Spatial Block-Correlation Model for Fluid Antenna Systems}, 
  year={2024},
  volume={23},
  number={11},
  pages={15829-15843},
  doi={10.1109/TWC.2024.3434509}}

\end{document}